\documentclass[12pt]{article}

\usepackage[margin=1in]{geometry}
\usepackage[T1]{fontenc}
\usepackage{lmodern}
\usepackage{amsmath,amsthm,amssymb,mathrsfs,mathtools}
\usepackage{booktabs}
\usepackage{microtype}
\usepackage[dvipsnames]{xcolor}
\usepackage{hyperref}
\usepackage{enumitem}

\hypersetup{colorlinks=true,linkcolor=blue,citecolor=blue,urlcolor=blue}

\newtheorem{theorem}{Theorem}[section]
\numberwithin{theorem}{section} \numberwithin{equation}{section}

\newtheorem{prop}[theorem]{Proposition}

\theoremstyle{remark}
\newtheorem{remark}[theorem]{Remark}

\numberwithin{equation}{section}
\newcommand{\dd}{\mathop{}\!\mathrm{d}}

\allowdisplaybreaks

\title{Photon spheres, elliptic surfaces, and Seiberg--Witten geometry}

\author{%
	Cordell Blankenship
	\\
	\small Dept.\!~of Mathematics \& Sciences,
	Roane State Community College, Harriman, TN 37748
	\\
	\small \texttt{blankenshipdc1@roanestate.edu}
	\\[0.75em]
	Andreas Malmendier%
	\thanks{A.M. acknowledges support from the Simons Foundation through
		grant no.~202367.}
	\\
	\small Dept.\!~of Mathematics \& Statistics,
	Utah State University, Logan, UT 84322
	\\
	\small \texttt{andreas.malmendier@usu.edu}
	\\[0.75em]
	Michael T. Schultz
	\\
	\small Dept.\!~of Mathematics,
	Virginia Tech, Blacksburg, VA 24060
	\\
	\small \texttt{michaelschultz@vt.edu}
}

\date{}

\begin{document}
	\maketitle
	
	\begin{abstract}
		We study null geodesics in Reissner--Nordstr\"om--de~Sitter spacetime and in
		its \(\omega=-2/3\) Kiselev deformation through the elliptic curves defined
		by the radial equation. The relation of these curves to the non-critical
		\(E_7\) Seiberg--Witten family yields differential equations for
		finite-distance light deflection, and explains the boundary source term
		through a suitable gauge of the Seiberg--Witten differential. This geometry
		separates the extremal point \(Q^2=M^2\) from the Argyres--Douglas point
		\(Q^2=9M^2/8\). We also determine the Kiselev critical and horizon loci, and
		show that, at fixed nonzero angular momentum, the Argyres--Douglas and
		ultracold loci meet only at zero Killing energy.
	\end{abstract}
	

	\section{Introduction}
	\label{sec:introduction}
	The appearance of elliptic functions in null-geodesic problems on black-hole
	backgrounds is classical. Such functions already occur in the work of Droste
	and Hagihara on Schwarzschild geodesics~\cite{Droste1917,Hagihara1931}, in
	Darwin's analysis of the strong-deflection regime~\cite{Darwin1959}, and in
	modern treatments based on Weierstrass uniformization~\cite{GibbonsVyska2012}.
	Lensing by charged black holes, and its weak- and strong-field limits, has been
	studied extensively; representative references are~\cite{EiroaEtAl2002,Sereno2004,Bozza2002,ZhaoTangHe2016}.
	For spacetimes that are not asymptotically flat, the observable bending angle
	also contains local endpoint angles, a point emphasized in the literature on the
	cosmological constant and finite-distance lensing~\cite{RindlerIshak2007,IshiharaEtAl2016}.
	Villanueva et al.~\cite{Villanueva:2013zta} studied the light deflection angle in
	Reissner--Nordstr\"om--anti-de~Sitter spacetimes from the point of view of
	elliptic functions.
	\par More recently, the deflection angle in Schwarzschild, Reissner--Nordstr\"om, and
	Reissner--Nordstr\"om--de~Sitter spacetimes has been studied using
	Picard--Fuchs differential equations, which arise in algebraic geometry for
	families of elliptic curves. Such equations form a linear holonomic system, and
	govern how the period integrals of the curve, and hence the geometric information
	they encode, vary with the underlying complex structure
	\cite{MR229641,MR233825,MR282990,deligne_equations_1970}.\footnote{The linear holonomic system is obtained from a flat connection on a vector bundle that has geometric origin. The associated
		vector bundle, whose fibers are the cohomology groups of the curves, is
		holomorphic and carries the Hodge filtration. It is equipped with a canonical
		flat connection, the Gauss--Manin connection, whose local flat sections
		determine the holonomic linear system; see loc.~cit. This is made explicit in
		Section~\ref{sec:normalfunction}.} This perspective is available because the
	equation of motion for null geodesics in the equatorial plane, which expresses
	the deflection angle through elliptic functions, takes the form
	\begin{equation}
		\label{eq-general_quartic}
		y^2 = q(x) \, , \quad y := \frac{\dd x}{\dd\phi} \, ,
	\end{equation}
	where $x$ is a suitably normalized radial variable, $\phi$ is the azimuthal angle, and $q(x)$ is a
	quartic polynomial whose coefficients are (dimensionless) rational functions of
	the physical parameters of the black hole. We call $q(x)$ the \emph{photon
		quartic}. On the one hand, the equation of motion for $\phi=\phi(x)$ is naturally
	related to elliptic functions determined by the photon quartic; on the other hand,
	the compactification of the locus of solutions of \eqref{eq-general_quartic} is a
	family of genus-one curves, parametrized by the coefficients of the photon
	quartic. Specifics for the Reissner--Nordstr\"om--de~Sitter black hole are given in
	\eqref{eq:dimensionless} and \eqref{eq:photon-polynomial-general}. The Jacobian of
	this genus-one curve is an elliptic curve whose periods determine the associated
	homogeneous Picard--Fuchs equation. Sasaki and Suzuki treated the Schwarzschild
	deflection integral as an open elliptic period and derived an inhomogeneous
	Picard--Fuchs equation~\cite{SasakiSuzuki2021}. Sasaki subsequently obtained the
	two-variable system for the Reissner--Nordstr\"om family and related its
	compatibility condition to an algebraic Painlev\'e~VI solution~\cite{Sasaki2024}.
	At the limiting charge $Q^2/M^2=9/8$, where the photon and anti-photon spheres
	coalesce, the strong-deflection expansion is non-logarithmic~\cite{Sasaki2025}.
	The present article identifies the global family of rational elliptic surfaces
	underlying these differential equations, and uses its geometry to interpret both
	their homogeneous and inhomogeneous parts.
	\par Periods of families of elliptic curves satisfy Picard--Fuchs
	equations~\cite{Griffiths1969,Stiller1981}; normal functions extend this
	period geometry to algebraic cycles and relative chains
	\cite{CarlsonHain1989}; and classical binary-quartic invariants recover
	the Jacobian of a genus-one quartic~\cite{Fisher2008}.  Rational elliptic
	surfaces with four singular fibers were classified by
	Herfurtner~\cite{Herfurtner1991}, and the
	particular family below occurs in Doran's study of algebraic
	isomonodromic deformations~\cite{Doran2001}.  Such surfaces also underlie
	rank-one Seiberg--Witten systems
	\cite{SeibergWitten1994,Shimizu2006,ArgyresLotitoLuMartone2016}.  Of
	special relevance is the trigonometric \(E_7\) curve of Minahan,
	Nemeschansky and Warner (MNW)
	\cite{MinahanNemeschanskyWarner1997}.  Its additive Kodaira degeneration
	provides a rank-one Argyres--Douglas point
	\cite{ArgyresDouglas1995}.
	\par Our main results are as follows.
	\begin{enumerate}
		\item The Jacobians of the genus-one curves \eqref{eq-general_quartic} defined by
		the photon quartics of Reissner--Nordstr\"om--de~Sitter spacetime form a
		one-parameter family of rational elliptic surfaces with generic
		configuration \(III^*+3I_1\).  Its Gauss--Manin connection is exactly the
		homogeneous system underlying the Reissner--Nordstr\"om equations
		of~\cite{Sasaki2024}. This is shown in Sections~\ref{sec:surface} and~\ref{sec:GMPF}.
		\item The coordinate deflection integral is a local lift of the normal
		function defined by the lensing divisor, as explained in Section~\ref{sec:normalfunction}.  A natural
		Seiberg--Witten primitive admits a horizontal gauge in which its
		Picard--Fuchs image is exactly the form whose endpoint pairing gives the
		inhomogeneous source. We explain this in detail in Section~\ref{sec:SW}.
		\item After a change of variables, the elliptic fibers agree with the
		equal-mass MNW \(E_7\) curve \eqref{eq:E7dict}, \eqref{eq:ec3}.  We compute the
		transformation of the two-form and the residues of the quartic primitive in
		Section~\ref{ssec:SWmassnormalization}, and compare the two induced special
		K\"ahler metrics in Section~\ref{ssec:SWmetrics}.
		\item The value \(s=1/4\) creates a semistable \(I_2\) fiber at
		\(t=0\) but does not merge the physical strong-field node, whereas
		\(s=9/32\) creates a type-\(II\) cusp and the critical exponent \(-1/6\). This
		analysis is carried out in Section~\ref{sec:enhancements}.  For the Kiselev
		deformation, we derive the full cusp locus in Section~\ref{sec:Kiselev_horizons_AD},
		pull back the two-variable Gauss--Manin system, and prove that the cusp
		and ultracold loci meet at fixed \(L\ne0\) only when \(E=0\).
	\end{enumerate}
	\par The term ``Kiselev quintessence'' is understood as in
	\cite{Kiselev2003}.  Its source is anisotropic and is not a quintessence
	scalar field in the standard cosmological sense~\cite{Visser2020}; we
	therefore use \emph{Kiselev deformation}.  Strong lensing in the
	\(\omega=-2/3\) neutral and charged Kiselev backgrounds was studied
	in~\cite{YounasEtAl2015,AzregAinouEtAl2017}.  The cold, Nariai and
	ultracold limits of charged de~Sitter black holes follow the standard
	classification~\cite{Romans1992,CastroMarianiToldo2023}.  We also use
	the established relation between circular null orbits and Lyapunov
	exponents~\cite{CardosoEtAl2009}.
	\begin{remark}
		\label{rmk:sw_curves}
		A similar analysis may be carried out for the elliptic integrals appearing in
		Reissner--Nordstr\"om--anti-de~Sitter spacetimes, as studied by Villanueva et
		al.~\cite{Villanueva:2013zta}. One recovers the same structure of elliptic
		fibrations, and a similar relation to Seiberg--Witten curves. For brevity, we omit
		this analysis; it may become the subject of future work. It is also of interest
		that the Seiberg--Witten curves of the classical $N_f=1,2$ cases of
		$\mathrm{SU}(2)$ supersymmetric gauge theory, as studied by Seiberg and
		Witten~\cite{SeibergWitten1994,MR1306869}, have appeared since 2020 in the
		perturbative analysis of quasinormal modes (QNM) of certain black-hole solutions
		\cite{MR4420567,MR4354399,MR4403520,MR4067647}. Most recently, this analysis has
		been extended to the extremal Reissner--Nordstr\"om--de~Sitter
		solution~\cite{Wang:2026xxx}. This has been termed the SW/QNM correspondence, and it
		forms an intriguing line of research that appears to be independent of the
		work undertaken here.
	\end{remark}
	%
	\section{From null geodesics to the photon curve}
	\label{sec:photon-curve}
	Consider the static and spherically symmetric metric
	\begin{equation}
		\dd s^2=-f(r)\,\dd T^2+\frac{\dd r^2}{f(r)}
		+r^2\bigl(\dd\theta^2+\sin^2\theta\,\dd\phi^2\bigr).
		\label{eq:metric}
	\end{equation}
	On the equatorial plane, the conserved energy and angular momentum of a
	null geodesic are $E=f(r)\dot T$ and $L=r^2\dot\phi$, where a dot denotes
	differentiation with respect to an affine parameter.  For \(L\ne0\), with
	\(U=1/r\), the null condition reduces to
	\begin{equation}
		\left(\frac{\dd U}{\dd\phi}\right)^2
		=\frac{E^2}{L^2}-U^2f(1/U)
		=\frac{1}{b^2}-U^2f(1/U)\quad(E\ne0),
		\label{eq:orbit-general}
	\end{equation}
	where \(b=L/E\) is the impact parameter when \(E\ne0\).
	\par For Reissner--Nordstr\"om--de~Sitter spacetime supplemented by a
	Kiselev term, take
	\begin{equation}
		f(r)
		=
		1-\frac{2M}{r}+\frac{Q^2}{r^2}
		-\frac{\gamma}{r^{1+3\omega}}
		-\frac{\Lambda r^2}{3}.
		\label{eq:Kiselev_f_general}
	\end{equation}
	Here \(M>0\), \(Q^2\geq 0\), and
	\(-1<\omega<-1/3\).  The stress tensor supporting the Kiselev
	term is anisotropic, although its angular average has equation of
	state \(p_q=\omega\rho_q\).  In the normalization in which Einstein's
	equation is \(G_{\mu\nu}=8\pi G T_{\mu\nu}\), its energy density is
	\begin{equation}
		8\pi G\,\rho_q(r)
		=
		-\frac{3\omega\gamma}{r^{3(1+\omega)}}.
		\label{eq:Kiselev_density}
	\end{equation}
	Across the Kiselev range \(-1<\omega<-1/3\), nonnegative energy density
	requires $\gamma\geq0$.  At the two endpoints of this range, the Kiselev term
	reproduces a cosmological-constant term and a constant term, respectively, and
	so does not introduce a genuinely new radial dependence.
	For \(\omega=-2/3\), we introduce
	\begin{equation}
		\begin{aligned}
			x&=2MU, &
			s&=\frac{Q^2}{4M^2}, &
			a&=2M\gamma, &
			\ell&=\frac{4M^2\Lambda}{3},\\
			t&=\frac{4M^2E^2}{L^2}+\ell
			=\frac{4M^2}{b^2}+\ell && (E\ne0).
		\end{aligned}
		\label{eq:dimensionless}
	\end{equation}
	Equation~\eqref{eq:orbit-general} becomes
	\begin{equation}
		\left(\frac{\dd x}{\dd\phi}\right)^2
		=q_{s,t,a}(x),
		\qquad
		q_{s,t,a}(x)=t+a x-x^2+x^3-sx^4.
		\label{eq:photon-polynomial-general}
	\end{equation}
	When its discriminant is nonzero, the compactification of
	\begin{equation}
		C_{s,t,a}:\qquad y^2=q_{s,t,a}(x)
		\label{eq:photon-curve_quintessence}
	\end{equation}
	is a genus-one curve with holomorphic differential
	\(\omega_C=\dd x/y\).  For \(a=0\) we write
	\begin{equation}
		C_{s,t}:\qquad y^2=q_{s,t}(x):=t-x^2+x^3-sx^4
		\label{eq:photon-curve}
	\end{equation}
	and first analyze this normalized family.  The linear Kiselev term can
	be removed:
	
	\begin{prop}
		\label{prop:linear-term-removal}
		Put
		\begin{equation}
			D=3a^2+8t,
			\qquad
			N=a^3+4at+8t^2,
			\qquad
			H=3a^4+16a^2t+64at^2+256st^3.
			\label{eq:DNH}
		\end{equation}
		On the open set $tDN\ne0$, $C_{s,t,a}$ is birational to
		\begin{equation}
			C_{\widetilde{s},\widetilde{t}}:\qquad
			\widetilde{y}^2=q_{\widetilde{s},\widetilde{t}}(\widetilde{x})
			=
			\widetilde{t}-\widetilde{x}^2+\widetilde{x}^3
			-\widetilde{s}\,\widetilde{x}^4,
			\label{eq:quartic_curve2}
		\end{equation}
		where
		\begin{equation}
			\widetilde t=\frac{8N^2}{D^3},
			\qquad
			\widetilde s=\frac{DH}{32N^2}.
			\label{eq:tilde-parameters}
		\end{equation}
		The holomorphic differentials satisfy
		\begin{equation}
			\frac{\dd x}{y}
			=\sqrt{\frac{8t}{3a^2+8t}}\,
			\frac{\dd\widetilde x}{\widetilde y},
			\label{eq:tilde-differential}
		\end{equation}
		with a consistent choice of square root. In particular,
		$(\widetilde s,\widetilde t)=(s,t)$ when $a=0$. 
	\end{prop}
	
	\begin{proof}
		On the open set $tDN\ne0$, define
		\begin{equation}
			A_0=\frac{tD}{N},
			\qquad
			B_0=-\frac{aD}{4N},
			\qquad
			C_0^2=\frac{tD^3}{8N^2},
			\label{eq:linear-removal-coefficients}
		\end{equation}
		and make the fractional-linear change of variables
		\begin{equation}
			x=\frac{A_0\widetilde x}{1+B_0\widetilde x},
			\qquad
			y=\frac{C_0\widetilde y}{(1+B_0\widetilde x)^2}.
			\label{eq:linear-removal-map}
		\end{equation}
		Multiplying $q_{s,t,a}(A_0\widetilde x/(1+B_0\widetilde x))$
		by $(1+B_0\widetilde x)^4/C_0^2$, the linear coefficient vanishes,
		the quadratic and cubic coefficients become $-1$ and $1$, and the
		constant and quartic coefficients are respectively
		$\widetilde t$ and $-\widetilde s$.  Differentiating
		\eqref{eq:linear-removal-map} gives
		\begin{equation}
			\frac{\dd x}{y}
			=\frac{A_0}{C_0}\frac{\dd\widetilde x}{\widetilde y}
			=\sqrt{\frac{8t}{3a^2+8t}}\,
			\frac{\dd\widetilde x}{\widetilde y}.
			\label{eq:tilde-differential_pf}
		\end{equation}
	\end{proof}
	We analyze the fixed-$s$ family first, and return to the pullback along~\eqref{eq:tilde-parameters} in
	Section~\ref{sec:Kiselev_horizons_AD}.
	\par Let $x_0$ be the turning root relevant to an exterior scattering orbit, and let $x_R=2M/r_R$ and $x_S=2M/r_S$.  The coordinate angular separation is the relative period
	\begin{equation}
		\phi_{RS}(s,t)=
		\int_{x_R}^{x_0}\frac{\dd x}{\sqrt{q_{s,t}(x)}}
		+\int_{x_S}^{x_0}\frac{\dd x}{\sqrt{q_{s,t}(x)}}.
		\label{eq:coordinate-angle}
	\end{equation}
	With the orientation convention used here, the finite-distance deflection angle is~\cite{IshiharaEtAl2016}
	\begin{equation}
		\alpha=\phi_{RS}+\Psi_R+\Psi_S-\pi,
		\qquad
		\sin\Psi_i=b\,U_i\sqrt{f(r_i)},
		\quad i=R,S.
		\label{eq:finite-distance-angle}
	\end{equation}
	In an asymptotically flat geometry, $x_R,x_S\to0$ and
	\begin{equation}
		\alpha=2\Delta\phi-\pi,
		\qquad
		\Delta\phi(s,t)=\int_0^{x_0}\frac{\dd x}{\sqrt{q_{s,t}(x)}}.
		\label{eq:asymptotic-angle}
	\end{equation}
	For $\Lambda\ne0$, the elliptic integral depends on $b$ and $\Lambda$ only through the combination $t$, but the endpoint angles in
	\eqref{eq:finite-distance-angle} retain their separate dependence. Thus the full finite-distance observable is not, in general, a function of $(s,t)$ alone.
	%
	\section{The family of rational elliptic surfaces}
	\label{sec:surface}
	
	Let $\hat{\pi} : \mathcal{C} \to B$ denote the family of genus-one curves obtained from \eqref{eq:photon-curve}, where $B=\mathbb P^1_{(t)}\times\mathcal M_{(s)}$ and $C_{s,t} = \hat{\pi}^{-1}(s,t)$. The family $\pi : \mathcal{E} \to B$ of Jacobian elliptic curves of the fibers of $\mathcal{C}$ is birationally equivalent
	to the Weierstrass model
	\begin{equation}
		E_{s,t}:\qquad \eta^2=4\xi^3-g_2(s,t)\,\xi-g_3(s,t),
		\qquad
		g_2=\frac43-16st,\quad g_3=\frac{32}{3}st+\frac{8}{27}-4t,
		\label{eq:weier}
	\end{equation}
	with holomorphic one-form $\dd\xi/\eta=\dd x/(2y)$. The birational map, valid for $t>0$ with $\tau=\sqrt t$, is
	\begin{equation}
		\xi^{(\tau)}=-\frac{2\tau y}{x^2}-\frac13+\frac{2t}{x^2},
		\qquad
		\eta^{(\tau)}=-\frac{8t}{x^3}y+2\tau-\frac{4\tau}{x}+\frac{8\tau t}{x^3},
		\label{eq:birational}
	\end{equation}
	(An analogous map defined over $\mathbb Q(i\sqrt s)$ exists for $s>0$.) For each fixed \(s\), the minimal compactification
	of~\eqref{eq:weier} over $\mathbb P^1_{(t)}$ is a rational elliptic
	surface~\cite{Miranda1989}. Allowing $s$ to vary gives a two-parameter
	elliptic fibration
	\begin{equation}
		\pi:\mathcal E\longrightarrow B=\mathbb P^1_{(t)}\times\mathcal M_{(s)}
		\label{eq:family}
	\end{equation}
	which may be thought of as a one-parameter family of rational elliptic surfaces, parameterized by the deformation parameter $s$. For $s$ fixed and generic, the elliptic fibration $\mathcal{E}$ is called the relative Jacobian fibration of the genus-one fibration $\mathcal{C}$. The two are birational as algebraic varieties, and the birational map \eqref{eq:birational}, restricted to each smooth fiber, is an isomorphism known as the Abel--Jacobi map. Working with $\mathcal{E}$ instead of $\mathcal{C}$ gives access to more of the relevant geometric structure, in particular the singular fibers, and allows the direct computation of the Picard--Fuchs operators, which is carried out in Section~\ref{sec:GMPF}.
	\par First, we analyze the structure of the singular fibers. Fix $s$. Then the singular fibers of the fibration $\pi : \mathcal{E} \to \mathbb{P}^1_{(t)}$ lie over the roots of the discriminant 
	\begin{equation}
		\label{eq:discr}
		\Delta(s,t) := g_2(s,t)^3 - 27g_3(s,t)^2 \, .
	\end{equation}
	Since \eqref{eq:birational} is the identity map on the base $B$, these are also the values of $t$ at which the genus-one curve \eqref{eq:photon-curve} becomes singular. At $t=0$, we have $g_2=4/3$ and $g_3=8/27$ in \eqref{eq:weier}, and
	\begin{equation}
		4\xi^3-\tfrac43\xi-\tfrac8{27}=\tfrac{4}{27}(3\xi-2)(3\xi+1)^2,
		\label{eq:nodalfiber}
	\end{equation}
	a nodal cubic.  For generic \(s\), the corresponding fiber of the
	elliptic surface is of type \(I_1\); at \(s=1/4\), the discriminant has a
	double zero and the fiber enhances to \(I_2\), as discussed in
	Section~\ref{sec:enhancements}.  Writing $s=(9-S^2)/32$, the remaining
	factor of $\Delta(s,t)$ splits, giving two further $I_1$ fibers at
	\begin{equation}
		t_c^{(\pm)}=\frac{8(S\pm1)}{(S\pm3)^3},
		\label{eq:tcpm}
	\end{equation}
	which correspond, on the quartic model, to the collision of two real branch
	points at $x_c^{(\pm)}=\pm4/(S\pm3)$. Thus a generic member of the family has
	exactly three finite $I_1$ fibers: the generic weak-field node at $t=0$, and
	the two photon-sphere nodes at $t_c^{(\pm)}$.
	\par Equation \eqref{eq:weier} is (a normalization of) a family of rational elliptic surfaces constructed by Herfurtner in his classification of elliptic surfaces with four singular fibers~\cite{Herfurtner1991}, where it appears with Kodaira type $III^*+3I_1$ and Mordell--Weil group generated by a single infinite-order section; the coordinate change
	\begin{equation}
		\widehat\xi=\frac{\widehat t^{\,2}\xi}{\rho},\quad
		\widehat\eta=\frac{\widehat t^{\,3}\eta}{\rho^{3/2}},\quad
		\widehat t=\frac{2}{9t(8s-3)},\quad
		\sigma=-\frac83\frac{s}{8s-3},\quad
		\rho=\frac23,
		\label{eq:tildecoords}
	\end{equation}
	converts it into the normal form
	\[
	\widehat\eta^{\,2}
	=4\widehat\xi^{\,3}
	-3\widehat t^{\,3}(\sigma+\widehat t)\widehat\xi
	-\widehat t^{\,5}(1+\widehat t),
	\]
	which was used in \cite{Doran2001}.  We use hats here to distinguish this modular
	normalization from the Kiselev variables
	\((\widetilde s,\widetilde t)\) in
	Proposition~\ref{prop:linear-term-removal}.  As \(\sigma\) varies, the
	family interpolates between the extremal modular surface for
	\(\Gamma_0(1)^*\) at \(s=0\), with configuration
	\(II^*+2I_1\), and the surface for \(\Gamma_0(2)\) at \(s=1/4\),
	with configuration \(III^*+I_2+I_1\).  The latter is precisely where
	the Mordell--Weil group acquires \(2\)-torsion.  The additional section
	of~\eqref{eq:weier},
	\[
	\xi=\frac23-\frac1{4s},
	\qquad
	\eta=\frac{i(4s-1)}{4s^{3/2}},
	\]
	is generically of infinite order and becomes \(2\)-torsion at
	\(s=1/4\) (and in the \(s\to\infty\) limit).
	\par Minahan, Nemeschansky and Warner give the Seiberg--Witten curve of the
	non-critical \(E_7\) string with two Wilson-line masses \(m_1,m_2\)
	as~\cite{MinahanNemeschanskyWarner1997}\footnote{Note that the variables $x,y$ here are unrelated to those in Section~\ref{sec:photon-curve}.}
	\begin{equation}
		y^2=x^3+\frac{u^2}4x^2-2u\Bigl(4\sin^2\bigl(\tfrac{m_+}2\bigr)x+u^2\Bigr)
		\Bigl(4\sin^2\bigl(\tfrac{m_-}2\bigr)x+u^2\Bigr),
		\qquad m_\pm=m_1\pm m_2.
		\label{eq:E7curve}
	\end{equation}
	Setting $m_-=0$ (equal masses) and $m_+=2m$, and applying the substitution
	\begin{equation}
		\xi=-\frac{4x}{u^2}-\frac13,\qquad \eta=\frac{16iy}{u^3},\qquad
		t=\frac{128\cos^2m}{u},\qquad s=-\frac14\tan^2m
		\label{eq:E7dict}
	\end{equation}
	to our Weierstrass model~\eqref{eq:weier} reproduces~\eqref{eq:E7curve} exactly at $m_-=0$:
	\begin{equation}
		y^2=x^3+\frac{u^2}4x^2-8\sin^2(m)\,u^3x-2u^5.
		\label{eq:ec3}
	\end{equation}
	Thus the Jacobian of \(C_{s,t}\), and hence the photon curve up to the
	birational identification above, is also the equal-mass non-critical \(E_7\)
	Seiberg--Witten curve.  Here \(u\) is the Coulomb-branch coordinate and
	\(m\) is the remaining Wilson-line mass.  Restricting further to
	\(s=0\) (\(m=0\)) gives the non-critical \(E_8\) curve.
	%
	\section{Gauss--Manin connection and Picard--Fuchs equation}
	\label{sec:GMPF}
	For fixed $s$, the periods $\Pi(t;s)=\int_{\gamma_t}d\xi/\eta$ of the rational elliptic surface~\eqref{eq:weier} satisfy a second-order linear ODE in $t$,
	\begin{equation}
		\mathcal L_t(s)\,\Pi(t;s):=\Bigl(\partial_t^2+P(t,s)\partial_t+Q(t,s)\Bigr)\Pi(t;s)=0.
		\label{eq:PFhom}
	\end{equation}
	Applying Stiller's reduction in \cite[\S 3]{Stiller1981} to the Weierstrass coefficients
	\eqref{eq:weier} gives
	\begin{equation}
		D_0=16s^2t-6st+4s-1,
		\qquad
		D_1=256s^3t^2+128s^2t-144st+16s+27t-4,
		\label{eq:PF-denominators}
	\end{equation}
	and then the coefficients
	\begin{equation}
		P(t,s)=\frac{2\mathcal N_P(t,s)}{tD_0D_1}, \quad
		Q(t,s)=\frac{3\mathcal N_Q(t,s)}{4tD_0D_1}.
		\label{eq:PQterm}
	\end{equation}
	Stiller's method is briefly reviewed in Appendix~\ref{ss-stiller}. The explicit polynomial numerators \(\mathcal N_P,\mathcal N_Q\) are
	listed in Appendix~\ref{ss:coeffs}.
	Equation~\eqref{eq:PFhom} can be solved in closed form by
	\begin{equation}
		\Pi(t;s)=g_3(t,s)^{-1/6}F(z),
		\qquad
		z=\frac{1}{27}\frac{g_2(t,s)^3}{g_3(t,s)^2},
		\qquad
		F(z)={}_2F_1\!\left(\tfrac1{12},\tfrac7{12};\tfrac23;z\right),
		\label{eq:hypergeom}
	\end{equation}
	using Gauss's hypergeometric function; a basis of solutions can then be constructed by standard methods.
	\par Allowing $s$ to vary, the periods no longer close under an ODE in $s$ alone; instead differentiating~\eqref{eq:hypergeom} in $t$ and $s$ and eliminating $F'(z)$ gives a first-order equation for $\partial_s\Pi$ in terms of $\Pi$ and $\partial_t\Pi$,
	\begin{equation}
		\widetilde{\mathcal L}_{t}(s)\,\Pi(t,s):=\Bigl(\partial_s+A(t,s)\partial_t+B(t,s)\Bigr)\Pi=0,
		\label{eq:GMeqn}
	\end{equation}
	\begin{equation}
		A(t,s)=-\frac{2t(8st-9t+2)}{16s^2t-6st+4s-1},
		\qquad
		B(t,s)=\frac{3t}{16s^2t-6st+4s-1}.
		\label{eq:ABcoeffs}
	\end{equation}
	Equations~\eqref{eq:PFhom} and~\eqref{eq:GMeqn} together give the so-called
	Gauss--Manin connection \(\nabla_{\mathrm{GM}}\) of the Jacobian elliptic curve family
	in the period basis \(\{\Pi,\partial_t\Pi\}\). More specifics will be given in Section \ref{sec:normalfunction}. Motion in \(t\) is the
	ordinary Picard--Fuchs variation at fixed \(s\), whereas motion in \(s\)
	deforms the elliptic surface; equation~\eqref{eq:GMeqn} expresses
	\(\partial_s\Pi\) in the same two-dimensional period basis.
	
	After \(s=Q^2/(4M^2)\) and \(t=4M^2/b^2\) (that is, \(\Lambda=0\)) are imposed,
	\eqref{eq:PQterm} agrees with Sasaki's charged deformation, and
	\eqref{eq:ABcoeffs} reproduces his isomonodromic equation
	\cite{Sasaki2024}.  Here both follow directly from the Weierstrass data
	\((g_2,g_3)\), independently of the deflection-angle calculation.
	\begin{prop}[Isomonodromy]
		\label{prop:isomonodromy}
		Away from collision values of \(s\), the deformation of
		\(\mathcal L_t(s)\) defined by~\eqref{eq:GMeqn} is isomonodromic.
		In particular, the monodromy representation of~\eqref{eq:PFhom} around
		the singular points \(t=0,t_c^{(\pm)}\) (of which \(t_c^{(\pm)}\) move with \(s\))
		is independent of \(s\), up to simultaneous conjugation.
	\end{prop}
	\begin{proof}
		The Gauss--Manin connection $\nabla_{\mathrm{GM}}$ is well-known to be flat on the complement of the
		discriminant locus.  Parallel transport in the \(s\)-direction identifies the
		homology local systems of nearby \(t\)-fibers and transports the punctured
		\(t\)-line and its loops.  Their monodromy representations are therefore
		conjugate.  In the explicit coordinates above, the flatness condition is
		the algebraically solvable Painlev\'e~VI system found in
		\cite{Sasaki2024}; equivalently, \(s\) is the isomonodromic time in
		Doran's description~\cite{Doran2001}.
	\end{proof}
	%
	\section{The deflection angle as lifted normal function}
	\label{sec:normalfunction}
	We now make precise the sense in which the finite-distance angular separation $\phi_{RS}(t,s)$  is a period-theoretic object. It is a so-called normal function, a notion developed by Griffiths~\cite{MR527828}.
	\par Over the complement $B^\circ=B\setminus\{\Delta=0\}$ of the
	discriminant locus in
	$B=\mathbb P^1_{(t)}\times\mathcal M_{(s)}$, the smooth fibers
	of~\eqref{eq:family} define the Gauss--Manin bundle
	\begin{equation}
		\mathbb H^1=R^1\pi_*\underline{\mathbb C}\otimes\mathcal O_{B^\circ},
		\qquad
		\nabla_{\mathrm{GM}}:\mathbb H^1\to\mathbb H^1\otimes\Omega^1_{B^\circ},
		\label{eq:GMsystem}
	\end{equation}
	with its canonically flat connection. Here $\underline{\mathbb C}$ is the constant sheaf on $\mathcal{E}^\circ:=\pi^{-1}(B^\circ)$, and we also write $\mathcal{C}^\circ:=\hat\pi^{-1}(B^\circ)$. For a locally flat basis of cycles
	$\Gamma_i \in H_1(C_{t,s},\mathbb Z)$, $i=0,1$, and a choice of nonvanishing holomorphic 1-form $\omega_{t,s} \in H^{1,0}(C_{t,s})$, the closed periods
	$\Pi_{\Gamma_i}(t,s)=\int_{\Gamma_i} \omega_{t,s}$ solve the homogeneous system
	attached to $\nabla_{\mathrm{GM}}$; see Section~\ref{sec:GMPF}.\footnote{The Gauss--Manin connection $\nabla_{\mathrm{GM}}$ associated to multiparameter families of elliptic curves  determines a differential geometric structure on $B^\circ$ known as a holomorphic projective structure. See Doran \cite[Sec. 2.2]{MR1779161}, Yoshida \cite[Ch. 8]{MR986252}.}
	\par Fix a generic $s$. On each smooth fiber $C_{t,s}$, we mark the points $P_T$ (the turning point, i.e.\ the branch point over the smallest positive root of $q_{s,t}$), $P_R$, and $P_S$ (receiver and source, lying over the finite values $x_R,x_S$). The relative one-chains $\Gamma_{R,T},\Gamma_{S,T}$ with $\partial\Gamma_{R,T}=P_T-P_R$ and $\partial\Gamma_{S,T}=P_T-P_S$ give
	\begin{equation}
		\Gamma_{RS}=\Gamma_{R,T}+\Gamma_{S,T},
		\qquad
		\partial\Gamma_{RS}=2P_T-P_R-P_S=:D_{\mathrm{lens}},
		\label{eq:lensdivisor}
	\end{equation}
	a divisor of degree zero on $C_{t,s}$. By definition,
	\begin{equation}
		\phi_{RS}(t,s)=\int_{\Gamma_{RS}}\omega_{t,s}.
	\end{equation}
	Since $D_{\mathrm{lens}}$ has degree zero, its Abel--Jacobi image
	\begin{equation}
		\nu_{\mathrm{lens}}(t,s)=\operatorname{AJ}_{C_{t,s}}(D_{\mathrm{lens}})=  \left[\int_{\Gamma_{RS}}\omega_{t,s}\right]\in\operatorname{Jac}(C_{t,s}) \cong E_{t,s}
		\label{eq:normalfunc}
	\end{equation}
	is well defined modulo the period lattice $\Lambda_{t,s} := \mathrm{span}_{\mathbb{Z}}\{\Pi_{\Gamma_0}(t,s),\Pi_{\Gamma_1}(t,s)\} \subset \mathbb{C}$ and defines a
	\emph{normal function}
	\begin{equation}
		\nu_{\mathrm{lens}}:B^\circ\longrightarrow J(\mathcal C^\circ/B^\circ) \cong \mathcal{E}^\circ
	\end{equation}
	valued in the relative Jacobian fibration, which is represented algebraically by the Weierstrass model \eqref{eq:weier} obtained from $\mathcal{C}$. The multivalued function
	$\phi_{RS}(t,s)$ is precisely a local lift of $\nu_{\mathrm{lens}}$ to the
	universal cover of $B^\circ$.  Thus it is a relative period even though it
	is neither single-valued nor the period of a closed cycle.
	\par The physical deflection angle $\alpha(t,s)$ differs from $\phi_{RS}$
	only by boundary terms.  Define
	\(\psi(P_i)=\arcsin\!\bigl(bU_i\sqrt{f(r_i)}\bigr)\) for
	\(i=R,S\), and \(\psi(P_T)=\pi/2\).  In relative de Rham cohomology,
	\(\alpha\) is the pairing of
	\(\widehat\omega_{t,s}=(\omega_{t,s},\psi_{t,s})\in
	H^1_{\mathrm{dR}}(C_{t,s},\{P_T,P_R,P_S\})\) with
	\(\Gamma_{RS}\):
	\begin{equation}
		\alpha(t,s)=\int_{\Gamma_{RS}}\omega_{t,s}-\psi_{t,s}(\partial\Gamma_{RS})
		=\phi_{RS}(t,s)+\psi(P_R)+\psi(P_S)-\pi.
		\label{eq:alpharelative}
	\end{equation}
	Both $\phi_{RS}$ and $\alpha$ are therefore relative periods attached to
	the same degree-zero divisor $D_{\mathrm{lens}}$.  They differ by an
	$H^0$-valued endpoint term.  This term does contribute to the
	Gauss--Manin derivative, but only through the
	elementary functions in~\eqref{eq:finite-distance-angle}; the genuinely
	elliptic inhomogeneity is the one derived in
	Section~\ref{ssec:inhomogeneous-deflection}.
	\par The role of the Weierstrass model~\eqref{eq:weier} is to make the Abel--Jacobi map explicit: uniformizing $\xi=\wp(w),\eta=\wp'(w)$ by the Weierstrass $\wp$-function of the period lattice $\Lambda_{t,s}$, with $w \in \mathbb{C}$, turns the normal function into a (multivalued) sum of $\wp^{-1}$'s,
	\begin{equation}
		\phi_{RS}=2\left(2w_T^{(\tau)}-w_R^{(\tau)}-w_S^{(\tau)}\right)
		=2\left(2\wp^{-1}(\xi^{(\tau)}_T)-\wp^{-1}(\xi^{(\tau)}_R)-\wp^{-1}(\xi^{(\tau)}_S)\right)
		\quad(\mathrm{mod\ periods}),
		\label{eq:phiRSwp}
	\end{equation}
	which we use in Section~\ref{sec:expansions} to compute the weak- and strong-field expansions directly from the local behavior of $\wp^{-1}$ near the degenerating fibers.
	%
	\section{Seiberg--Witten differential and source term}
	\label{sec:SW}
	
	Equations~\eqref{eq:PFhom} and~\eqref{eq:GMeqn} govern closed
	periods.  The deflection angle is instead a relative period and therefore
	acquires endpoint sources.  We now show that these sources are not external
	to the elliptic-surface geometry: after a suitable gauge choice, they are
	obtained by applying the same Picard--Fuchs operator to a
	Seiberg--Witten primitive of the natural two-form.
	
	\subsection{Meromorphic differentials on the quartic}
	\label{ssec:meromorphic-differentials}
	
	Let $s\ne0$ and write a smooth genus-one curve $C_{s,t}$ \eqref{eq:photon-curve} as
	$y^2=f(x)=-s\prod_{j=1}^4(x-e_j)$.  For
	$\eta_i=\dd x/((x-e_i)y)$, direct differentiation gives
	\begin{equation}
		\frac{\dd x}{y}
		=-\frac12\sum_{i=1}^4e_i\eta_i
		-\dd\!\left(\frac{x}{y}\right).
		\label{eq:decomposition}
	\end{equation}
	In the coordinate $z=\sqrt{x-e_i}$ at $P_i=(e_i,0)$,
	\begin{equation}
		\eta_i
		=\frac{2}{\sqrt{f'(e_i)}}\frac{\dd z}{z^2}
		+O(1)\,\dd z.
		\label{eq:etalocal}
	\end{equation}
	Thus each $\eta_i$ is of the second kind.  The quartic compactification has
	two points
	\[
	C_\pm:\quad x=\infty,\qquad y/x^2=\pm\sqrt{-s}.
	\]
	Using $\zeta=1/x$, one finds at $\zeta =0$ that
	\begin{equation}
		\operatorname{Res}_{C_\pm}\frac{x\,\dd x}{y}
		=\mp\frac1{\sqrt{-s}},
		\qquad
		\operatorname{Res}_{C_\pm}\frac{x^2\,\dd x}{y}
		=\pm\frac1{2(-s)^{3/2}}.
		\label{eq:meromorphic-basis-residues}
	\end{equation}
	Consequently $\dd x/y$, $x\,\dd x/y$, and $x^2\,\dd x/y$ form a
	convenient meromorphic reduction basis on the curve punctured at
	$C_\pm$, but they are not all differentials of the second kind.  This
	distinction is essential when the Seiberg--Witten residues are compared
	below.
	
	\subsection{The Seiberg--Witten differential}
	\label{ssec:SWdiff}
	
	Seiberg--Witten integrable systems were studied in general by Shimizu~\cite[Def.~3.1.4]{Shimizu2006}, building on earlier work of Donagi~\cite{MR1726927}. They are specified by the choice of a rational elliptic surface $\mathcal{E} \to \mathbb{P}^1_{(t)}$ together with a holomorphically varying rational 1-form $\lambda$ on each smooth fiber $E_t$, such that $\partial_t\lambda = \omega \in H^{1,0}(E_t)$ is the nonvanishing holomorphic 1-form and $\dd \lambda = \Omega$ is a holomorphic 2-form on the complement of the singular fibers in the total space. The former identity is the defining Gauss--Manin property: differentiating a
	period of $\lambda$ in the base direction gives a period of the relative
	holomorphic differential.  These equations are collectively known as the Seiberg--Witten equations. 
	\par In this situation, we work with the genus-one fibration $\mathcal{C}$ rather than the relative Jacobian fibration $\mathcal{E}$. This is permissible because the generic fiber of $\mathcal{C}$ has a rational point over the function field of $\mathbb{P}^1_{(t)}$, for instance the points $C_\pm$ at infinity. At fixed $s$, put
	\begin{equation}
		\Omega_C=\frac{\dd t\wedge\dd x}{y} \in \Omega^2(\mathcal{C}^\circ) \, ,
		\label{eq:quartic-two-form}
	\end{equation}
	where $\mathcal{C}^\circ=\hat\pi^{-1}(B^\circ)$ is the smooth locus of the fibration \eqref{eq:family} at the fixed value of $s$. The quartic admits the algebraic Seiberg--Witten primitive
	\begin{equation}
		\lambda_C=2y\,\dd x,
		\qquad
		\dd\lambda_C
		=2\,\dd y\wedge\dd x
		=\frac{\dd t\wedge\dd x}{y}
		=\Omega_C,
		\qquad
		\partial_t\lambda_C=\frac{\dd x}{y} \, ,
		\label{eq:SWgauge}
	\end{equation}
	which are the Seiberg--Witten equations for $\mathcal{C}$ at the fixed value of $s$.  A direct expansion at the two horizontal
	sections $C_\pm$ gives
	\begin{equation}
		\operatorname{Res}_{C_\pm}(\lambda_C)
		=\pm\frac{\sqrt{-s}\,(1-4s)}{8s^3}.
		\label{eq:SWresidues}
	\end{equation}
	The poles therefore lie along marked sections of the elliptic surface, not
	along a finite collection of singular fibers.  Their residues are independent
	of $t$, as required at fixed mass parameter.  The relation with the
	mass-normalized one-form of Minahan--Nemeschansky--Warner is addressed in
	Section~\ref{ssec:SWmassnormalization}.
	
	\subsection{Algebraic gauge reduction}
	\label{ssec:algebraic-gauge}
	
	For $s\ne0$, add the exact differential of
	\begin{equation}
		\chi_{\mathrm{red}}=-\frac23xy+\frac{y}{6s}
	\end{equation}
	to~\eqref{eq:SWgauge}.  The resulting representative is
	\begin{equation}
		\Lambda_C:=\lambda_C+\dd\chi_{\mathrm{red}}
		=-\frac{8s-3}{12s}\frac{x^2\dd x}{y}
		-\frac1{6s}\frac{x\,\dd x}{y}
		+\frac{4t}{3}\frac{\dd x}{y}
		-\left(\frac{x}{3y}-\frac1{12sy}\right)\dd t.
		\label{eq:Lambdadef}
	\end{equation}
	Writing
	\begin{equation}
		\Lambda_C=\mathfrak a(x,t,s)\,\dd x
		-\mathfrak b(x,t,s)\,\dd t,
		\qquad
		\mathfrak a
		=-\frac{8s-3}{12s}\frac{x^2}{y}
		-\frac1{6s}\frac{x}{y}
		+\frac{4t}{3y},
		\qquad
		\mathfrak b=\frac{x}{3y}-\frac1{12sy},
		\label{eq:abdef}
	\end{equation}
	exhibits the algebraic primitive in the meromorphic basis of
	Section~\ref{ssec:meromorphic-differentials}.  Since the gauge correction is
	exact,
	\begin{equation}
		\dd\Lambda_C=\Omega_C,
		\qquad
		\operatorname{Res}_{C_\pm}\Lambda_C
		=\operatorname{Res}_{C_\pm}\lambda_C.
	\end{equation}
	At $s=0$ this particular representative is singular, but the form-level
	object below is regular and gives the reduction directly.
	
	\subsection{The form-level identity and the Seiberg--Witten gauge}
	\label{ssec:SW-form-level}
	
	Let the Picard--Fuchs operator $\mathcal L_t(s)$ act coefficient-wise on forms.  Since its
	coefficients depend on $t$, it does not commute with $\partial_t$.  From
	\eqref{eq:abdef} one obtains
	\begin{equation}
		\mathcal L_t(s)\Omega_C
		=\dd  \left(\mathcal L_t(s)\Lambda_C\right)
		-\left(\partial_t\mathcal L_t(s)\right)\mathfrak a\,
		\dd t\wedge\dd x,
		\label{eq:formlevelPF}
	\end{equation}
	where $\partial_t\mathcal L_t$ differentiates only the coefficients of the
	operator.  The equivalent Griffiths reduction is more economical.  With
	$K=1/y$,
	\begin{equation}
		\mathcal L_t(s)K=\partial_xG,
		\qquad
		G=R_1K+R_3K^3,
		\label{eq:KPF}
	\end{equation}
	where $R_1,R_3$ are recorded in Appendix~\ref{ss:coeffs}.  Therefore
	\begin{equation}
		\mathcal L_t(s)\Omega_C=-\dd(G\,\dd t).
		\label{eq:omegaPF}
	\end{equation}
	
	The connection with the Seiberg--Witten gauge can be stated intrinsically.
	
	\begin{prop}
		\label{prop:horizontal-SW-gauge}
		At fixed $s$, on a simply connected chart of the smooth total space,
		there is a gauge
		representative of the associated Seiberg--Witten equations
		\begin{equation}
			\lambda_{\mathrm h}=-\mathscr A(x,t,s)\,\dd t,
			\qquad
			\partial_x\mathscr A=K,
			\qquad
			\dd\lambda_{\mathrm h}=\Omega_C,
			\label{eq:horizontal-SW}
		\end{equation}
		which differs locally from $\lambda_C$ by an exact form.  Its
		base-dependent gauge ambiguity can be fixed so that
		\begin{equation}
			\mathcal L_t(s)\lambda_{\mathrm h}=-G\,\dd t,
			\qquad
			\mathcal L_t(s)\dd\lambda_{\mathrm h}=-\dd(G\,\dd t).
			\label{eq:SW-horizontal-PF}
		\end{equation}
	\end{prop}
	
	\begin{proof}
		Choose a local Abelian integral $\mathscr A$ with
		$\partial_x\mathscr A=K$.  Then
		$\dd(-\mathscr A\,\dd t)=K\,\dd t\wedge\dd x=\Omega_C$, so
		$\lambda_{\mathrm h}-\lambda_C$ is locally exact.  Applying
		$\mathcal L_t(s)$ to $\partial_x\mathscr A=K$ and using
		\eqref{eq:KPF} shows that
		$\partial_x(\mathcal L_t\mathscr A-G)=0$.  The remaining function of
		$(t,s)$ can be removed by adding a base-dependent exact form to
		$\lambda_{\mathrm h}$.  This gives~\eqref{eq:SW-horizontal-PF}.
	\end{proof}
	
	This proposition is the precise form of the Seiberg--Witten interpretation
	of the inhomogeneity.  The one-form $G\,\dd t$ is not itself a
	Seiberg--Witten differential and it does not change the residues of
	$\lambda_C$; rather, it is (up to sign) the image of a Seiberg--Witten primitive under
	the Picard--Fuchs operator in the horizontal gauge.  Its endpoint values
	are consequently the source terms for relative periods.
	
	\subsection{Normalization and mass residues in the \texorpdfstring{$E_7$}{E7} variables}
	\label{ssec:SWmassnormalization}
	
	An isomorphism of elliptic fibrations does not by itself identify their
	Seiberg--Witten one-forms. Indeed, the Seiberg--Witten equations define a sheaf-theoretic condition that determines a special K\"ahler metric~\cite{Freed1999} in terms of the monodromy of the singular fibers and the residues of the Seiberg--Witten one-forms on the base of the fibration~\cite[\S 3.1]{SeibergWitten1994}. This metric is computed in Section~\ref{ssec:SWmetrics}. 
	\par To compare the normalizations, write the photon quartic \eqref{eq:photon-curve}
	in coordinates $(z,w)$,
	\begin{equation}
		w^2=t-z^2+z^3-sz^4,
		\qquad
		\lambda_C=2w\,dz,
		\qquad
		\dd\lambda_C=\frac{\dd t\wedge\dd z}{w},
		\label{eq:quarticSWrenamed}
	\end{equation}
	and reserve $(x,y,u)$ for~\eqref{eq:ec3}.  At fixed mass $m$, the
	birational map and~\eqref{eq:E7dict} give the relative identity
	\begin{equation}
		\frac{\dd z}{2w}
		=\frac{\dd\xi}{\eta}
		=\frac{iu}{4}\frac{\dd x}{y},
		\qquad\text{hence}\qquad
		\frac{\dd z}{w}=\frac{iu}{2}\frac{\dd x}{y}.
		\label{eq:relativeSWtransform}
	\end{equation}
	Together with
	\begin{equation}
		t=\frac{128\cos^2m}{u},
		\qquad
		\dd t=-\frac{128\cos^2m}{u^2}\,\dd u,
	\end{equation}
	and denoting the quartic primitive by
	\(\widetilde\lambda_C\), this yields
	\begin{equation}
		\dd\widetilde\lambda_C
		=\rho(u,m)\,\Omega_{\mathrm{MNW}},
		\qquad
		\rho(u,m)=\frac{64i\cos^2m}{u},
		\qquad
		\Omega_{\mathrm{MNW}}=\frac{\dd x}{y}\wedge\dd u.
		\label{eq:SWtwoformcomparison}
	\end{equation}
	The factor $\rho$ is due to both the inversion of the base and the
	$u$-dependent rescaling of the relative differential.  It can be absorbed
	locally by $v=64i\cos^2m\log u$, but this coordinate is multivalued.
	Accordingly, $\lambda_C$ is naturally normalized over the quartic
	$t$-coordinate, not over the physical MNW Coulomb coordinate $u$.
	
	The residue data give an independent obstruction.  With
	$s=-\tfrac14\tan^2m$ and
	$\sqrt{-s}=\tfrac12\tan m$, equation~\eqref{eq:SWresidues} becomes
	\begin{equation}
		\operatorname{Res}_{C_\pm}(\lambda_C)
		=\mp\frac{4\cos^3m}{\sin^5m}.
		\label{eq:quarticResidueMass}
	\end{equation}
	These residues are nonlinear in the additive Wilson-line mass and diverge
	as $\mp4m^{-5}$ when $m\to0$.  By contrast, the mass-normalized MNW
	differential has
	\begin{equation}
		\operatorname{Res}_{D_\alpha}
		\lambda_{\mathrm{MNW}}
		=
		\sum_i q_{\alpha i}m_i,
		\label{eq:MNWlinearResidues}
	\end{equation}
	with coefficients fixed by the flavor weights
	\cite[Appendix~A.2]{MinahanNemeschanskyWarner1997}.
	The logarithms in the MNW construction lift the periodic trigonometric
	dependence of the curve to the universal cover of mass space and thereby
	retain this additive information.
	
	An exact meromorphic correction cannot change~\eqref{eq:quarticResidueMass},
	because exact differentials have zero residues.  Passing to the
	mass-normalized form requires a closed, globally non-exact differential of
	the third kind with the appropriate pole divisor.  The relevant comparison
	is with the full two-mass differential of
	\cite[Appendix~A.2]{MinahanNemeschanskyWarner1997}, specialized to
	$m_-=0$ only after its spinor and null contributions have been combined.
	Thus the quartic and cubic equations identify the elliptic fibers, while
	\eqref{eq:SWtwoformcomparison} and~\eqref{eq:quarticResidueMass} measure the
	two independent differences between the corresponding normalized
	Seiberg--Witten geometries.
	
	\subsection{The induced special K\"ahler metrics}
	\label{ssec:SWmetrics}
	
	The distinction between the two Seiberg--Witten normalizations is also
	visible in the special K\"ahler metric.  Choose a locally constant
	symplectic basis \(A,B\in H_1(C_u,\mathbb Z)\), \(A\cdot B=1\), and set
	\begin{equation}
		\Pi_A=\oint_A\frac{\dd x}{y},
		\qquad
		\Pi_B=\oint_B\frac{\dd x}{y},
		\qquad
		\tau_{\rm ell}=\frac{\Pi_B}{\Pi_A}.
		\label{eq:cubicperiods}
	\end{equation}
	To avoid confusion with the Kiselev parameter, denote the special
	coordinates associated with any Seiberg--Witten differential \(\lambda\) by
	\(a_{\rm sp}=\oint_A\lambda\) and
	\(a_{D,\rm sp}=\oint_B\lambda\).  The rank-one metric is, up to a fixed
	positive convention-dependent constant \(c_{\rm SK}\),
	\begin{equation}
		g_{u\bar u}
		=
		c_{\rm SK}\,\operatorname{Im}\tau_{\rm eff}
		\left|\frac{\dd a_{\rm sp}}{\dd u}\right|^2,
		\qquad
		\tau_{\rm eff}=\frac{\dd a_{D,\rm sp}}{\dd a_{\rm sp}}.
		\label{eq:SKmetricgeneral}
	\end{equation}
	This is the standard rigid special K\"ahler construction
	\cite{Freed1999,ArgyresLotitoLuMartone2016}.
	
	For the transported quartic primitive, the two-form comparison
	\eqref{eq:SWtwoformcomparison} gives
	\begin{equation}
		\frac{\dd a_C}{\dd u}=-\rho\,\Pi_A,
		\qquad
		\frac{\dd a_{D,C}}{\dd u}=-\rho\,\Pi_B,
		\qquad
		\rho(u,m)=\frac{64i\cos^2m}{u}.
		\label{eq:quarticperiodderivativesu}
	\end{equation}
	The sign is immaterial for the metric.  For the MNW differential, write
	its Gauss--Manin normalization as
	\begin{equation}
		\nabla_{\partial_u}[\lambda_{\rm MNW}]
		=
		\kappa_{\rm MNW}\left[\frac{\dd x}{y}\right].
		\label{eq:MNWGMnormalization}
	\end{equation}
	Fiberwise exact terms drop out of periods, and hence
	\begin{equation}
		\frac{\dd a_{\rm MNW}}{\dd u}
		=\kappa_{\rm MNW}\Pi_A,
		\qquad
		\frac{\dd a_{D,\rm MNW}}{\dd u}
		=\kappa_{\rm MNW}\Pi_B.
		\label{eq:MNWperiodderivatives}
	\end{equation}
	It follows immediately that the two descriptions have the same effective
	coupling,
	\begin{equation}
		\tau_C=\tau_{\rm MNW}=\tau_{\rm ell}
		=\frac{\Pi_B}{\Pi_A},
	\end{equation}
	and their metrics are
	\begin{equation}
		g^{(C)}_{u\bar u}
		=
		c_{\rm SK}|\rho|^2\operatorname{Im}\tau_{\rm ell}\,|\Pi_A|^2,
		\qquad
		g^{({\rm MNW})}_{u\bar u}
		=
		c_{\rm SK}|\kappa_{\rm MNW}|^2
		\operatorname{Im}\tau_{\rm ell}\,|\Pi_A|^2.
		\label{eq:two-SK-metrics}
	\end{equation}
	Thus equality of the elliptic fibers fixes the period ratio, whereas the
	normalization of the holomorphic symplectic form fixes the special
	K\"ahler scale.  Locally the factor \(\rho\) may be absorbed into
	\(v=64i\cos^2m\log u\), but \(v\) is generally multivalued and is not the
	physical Coulomb coordinate used in the MNW normalization.
	\subsection{The inhomogeneous equation for the deflection angle}
	\label{ssec:inhomogeneous-deflection}
	
	The form-level identity becomes an equation for the lensing normal function
	only after the moving branch-point endpoint is treated correctly.  Set
	\(\Delta\phi(t,s)=\int_0^{x_0(t,s)}K\,\dd x\), where \(x_0\) is the
	turning root.  Although \(K\) diverges as
	\((x_0-x)^{-1/2}\), the divergences produced by differentiating the
	moving endpoint cancel those of the Griffiths certificate.  Equivalently,
	one may replace \(x_0\) by \(x_0-\varepsilon\), differentiate the
	regulated integral, combine all terms, and then take
	\(\varepsilon\to0^+\), as in~\cite[Sec.~3.1]{Sasaki2024}.
	
	\begin{prop}[Gauss--Manin equations for the lensing normal function]
		\label{prop:lensing-GM-system}
		The lifted normal function \(\Delta\phi\) satisfies
		\begin{equation}
			\mathcal L_t(s)\Delta\phi
			=
			-\frac{
				8s^3(8s-3)t^2
				+4s(13s-3)(4s-1)t
				-(4s-1)^2
			}{
				t^{3/2}
				(16s^2t-6st+4s-1)
				(256s^3t^2+128s^2t-144st+16s+27t-4)
			} ,
			\label{eq:GMinhom1}
		\end{equation}
		and
		\begin{equation}
			\widetilde{\mathcal L}_{t,s}\Delta\phi
			=
			-\frac{2\sqrt t}{16s^2t-6st+4s-1}.
			\label{eq:GMinhom2}
		\end{equation}
		These are respectively Eqs.~(3.14) and~(3.24) of
		\cite{Sasaki2024}, in the variables used here.
	\end{prop}
	
	For the second equation one uses
	\begin{equation}
		\widetilde{\mathcal L}_{t,s}K
		=
		\partial_x(\widetilde R_1K),
		\qquad
		\widetilde R_1
		=
		\frac{-x^2+2(1+4st-3t)x+2t}
		{16s^2t-6st+4s-1}.
	\end{equation}
	Applying this identity to
	\(\partial_x\mathscr A=K\) gives
	\(\partial_x(\widetilde{\mathcal L}_{t,s}\mathscr A
	-\widetilde R_1K)=0\).  The corresponding base normalization of the
	Abelian integral therefore yields, locally,
	\begin{equation}
		\widetilde{\mathcal L}_{t,s}\mathscr A=\widetilde R_1K,
		\qquad
		\widetilde{\mathcal L}_{t,s}\lambda_{\rm h}
		=-\widetilde R_1K\,\dd t.
		\label{eq:SW-horizontal-GM}
	\end{equation}
	Its compatibility with the \(t\)-component is the flatness condition of
	the two-variable Gauss--Manin system.
	For finite source and receiver, the same regulated argument yields the
	corresponding certificate at every endpoint; the elementary endpoint
	angles in~\eqref{eq:finite-distance-angle} must then be differentiated
	separately.
	
	Proposition~\ref{prop:lensing-GM-system} now follows from the
	Seiberg--Witten form argument rather than from an unrelated boundary
	calculation.  The homogeneous operators are the Gauss--Manin connection
	of the elliptic surface.  For the \(t\)-component,
	Proposition~\ref{prop:horizontal-SW-gauge} identifies the source as the
	regularized endpoint pairing of
	\(-\mathcal L_t\lambda_{\rm h}=G\,\dd t\); the
	\(\widetilde{\mathcal L}_{t,s}\)-component follows in the same way
	from~\eqref{eq:SW-horizontal-GM}.
	Thus the system is the relative-cohomology boundary of the
	Gauss--Manin action on a Seiberg--Witten primitive.
	
	\section{Weak-field and strong-field expansions}
	\label{sec:expansions}
	Every nodal fiber of the family carries two canonical asymptotic expansions of a relative period, depending on how the endpoints of the integration chain are scaled relative to the shrinking neighborhood of the node.
	\begin{prop}[Two asymptotic regimes]
		\label{prop:twoexpansions}
		Let $t\to t_c$ be an $I_1$ degeneration of the fiber at $x=x_c$, with
		local model
		\(q(x,t)=A(x-x_c)^2-(t_c-t)+\cdots\), \(A\ne0\).
		\begin{enumerate}[label=(\roman*)]
			\item (Vanishing-cycle scaling.) If both endpoints are scaled into the
			shrinking neighborhood, $x-x_c=\sqrt{t_c-t}\,X$ with $X$ fixed, the
			relative period tends to a finite limit, an elementary function (inverse
			trigonometric or inverse hyperbolic) of the rescaled endpoints.
			\item (Relative period.) If \(A>0\), and one endpoint is held fixed away from $x_c$
			while the other (the turning point) approaches the node as $t\to t_c^-$, the relative period diverges as
			\[
			-\frac{1}{2\sqrt A}\log(t_c-t)+O(1).
			\]
		\end{enumerate}
	\end{prop}
	Both statements follow from the same local Picard--Lefschetz monodromy of the $I_1$ fiber, evaluated on two different families of chains.
	\par Let us first consider the \emph{weak-field expansion} around the $I_1$ fiber, for $t\to0$. With $\tau=\sqrt t$, the degenerate cubic~\eqref{eq:nodalfiber} is parameterized by $\wp_0(w)=-\tfrac13+\csc^2w$, or, equivalently, $w=\wp_0^{-1}(\xi)=\arcsin\bigl((\xi+ \tfrac13)^{-1/2}\bigr)$. The turning point has the expansion
	\begin{equation}
		x_T=\tau+\tfrac12\tau^2+\tfrac{5-4s}{8}\tau^3+\tfrac{2-3s}{2}\tau^4+O(\tau^5),
		\label{eq:xTweak}
	\end{equation}
	giving $\xi_T^{(\tau)}=\tfrac53-2\tau+(2s-1)\tau^2+O(\tau^3)\to\tfrac53$ and $w_T^{(\tau)}=\tfrac\pi4+\tfrac\tau2+O(\tau^2)$.  In the
	asymptotically flat finite-distance scaling, set
	\(x_i=\tau\chi_i\) and \(x=\tau X\), keeping
	\(\chi_i=b/r_i\) fixed.  Then
	\begin{equation}
		\begin{aligned}
			2w_i^{(\tau)}
			&=\int_0^{\chi_i}
			\frac{\dd X}{\sqrt{1-X^2+\tau X^3-s\tau^2X^4}}\\
			&=\arcsin\chi_i
			-\frac\tau2\left(\frac1{\chi_i'}+\chi_i'-2\right)
			+O(\tau^2),
			\qquad \chi_i'=\sqrt{1-\chi_i^2}.
		\end{aligned}
		\label{eq:Wiweak}
	\end{equation}
	consistent with Proposition~\ref{prop:twoexpansions}(i): the leading term is the elementary period of the vanishing cycle, and the correction is regular in $\tau=\sqrt t$.  Using~\eqref{eq:phiRSwp} then gives the following weak-field expansion of the coordinate angular integral
	\begin{equation}
		\phi_{RS}=\pi-\arcsin\chi_R-\arcsin\chi_S
		+\tau\left[\frac{2-\chi_R^2}{2\sqrt{1-\chi_R^2}}+\frac{2-\chi_S^2}{2\sqrt{1-\chi_S^2}}\right]
		+O(\tau^2).
		\label{eq:phiRSweak}
	\end{equation}
	\par Let us also consider the \emph{strong-field expansion} around the $I_1$ fiber, for $t\to t_c^{(+)}$. Near $t_c^{(+)}$ the same local analysis as before applies with $x_c^{(+)}$ in place of the origin; because the receiver and source now sit at fixed, generic values of $x$ away from the node, case~(ii) of Proposition~\ref{prop:twoexpansions} is relevant. Concretely, near $x_0\to x_c^{(+)}$, we have
	\begin{equation}
		\begin{aligned}
			I_{\mathrm{sing}}(t)
			:=\int_{x_R}^{x_0}\frac{\dd x}{\sqrt{q_{s,t}(x)}}
			&=
			\frac1{\sqrt A}\,
			\operatorname{arcosh}
			\left(
			\delta\sqrt{\frac{A}{t_c^{(+)}-t}}
			\right)+O(1)
			\\
			&=
			-\frac1{2\sqrt A}\log(t_c^{(+)}-t)+O(1),
		\end{aligned}
		\label{eq:strongfield}
	\end{equation}
	where \(A=\tfrac12q_{xx}(x_c^{(+)},t_c^{(+)})>0\) and
	\(\delta>0\) is a fixed local cutoff.  This gives the logarithmic
	divergence as the impact parameter approaches its photon-sphere value
	\cite{Darwin1959,Bozza2002}.
	\par Both expansions originate in the same local model of a nodal
	degeneration.  In the weak-field limit the receiver and turning point are
	scaled together into the node; in the strong-field limit the receiver is
	held fixed while only the turning point approaches it.  The two classical
	lensing regimes are therefore the two canonical relative-period
	asymptotics associated with a Kodaira \(I_1\) degeneration.
	%
	\section{Fiber collisions and degeneration locus}
	\label{sec:enhancements}
	
	Two collisions of finite singular fibers occur at distinguished values of
	\(s\), but their local geometries and their lensing consequences are
	different.
	
	\begin{prop}
		\label{prop:I2-collision}
		At \(s=\tfrac14\), the two components of the discriminant locus meet at
		\(t=0\), and the fiber there is of Kodaira type \(I_2\).  The colliding vanishing
		cycles are mutually local.
	\end{prop}
	
	Indeed,
	\begin{equation}
		q_{1/4,0}(x)=-\frac14x^2(x-2)^2.
		\label{eq:I2-factorization}
	\end{equation}
	The two rational components meet at the two nodes \(x=0,2\), and the
	monodromy is conjugate to
	\(T^2=\left(\begin{smallmatrix}1&2\\0&1\end{smallmatrix}\right)\).
	This is a multiplicative, semistable enhancement.  It does not turn the
	physical strong-deflection singularity into a double logarithm: at
	\(s=\tfrac14\) that singularity remains the separate \(I_1\) fiber at
	\(t=\tfrac14\).  A relative chain probing only one of the two nodes in
	\eqref{eq:I2-factorization} has the ordinary local logarithm; the doubled
	monodromy records the sum of the two nodal contributions.
	
	\begin{prop}
		\label{prop:II-collision}
		At
		\begin{equation}
			(s_{\rm AD},t_{\rm AD})
			=\left(\frac9{32},\frac8{27}\right),
			\label{eq:canonical-AD-st}
		\end{equation}
		two \(I_1\) fibers with mutually non-local vanishing cycles collide to
		form a Kodaira \(II\) fiber.
	\end{prop}
	
	At this point
	\begin{equation}
		q_{9/32,\,8/27}(x)
		=
		-\frac9{32}
		\left(x-\frac43\right)^3
		\left(x+\frac49\right).
		\label{eq:II-factorization}
	\end{equation}
	Thus \(x_c=4/3\) is a triple root.  With
	\(\epsilon=t_{\rm AD}-t\) and
	\(x-x_c=\epsilon^{1/3}X\), the relative holomorphic differential and the
	quartic Seiberg--Witten differential scale as
	\begin{equation}
		\frac{\dd x}{y}
		\sim
		\epsilon^{-1/6}
		\frac{\dd X}{\sqrt{c_3X^3-1}},
		\qquad
		\oint_{\gamma_{\rm van}}\lambda_C
		\sim\epsilon^{5/6},
		\label{eq:IItype}
	\end{equation}
	where \(c_3\ne0\).  The first relation gives the fractional
	\(\epsilon^{-1/6}\) divergence of a lensing relative period that reaches
	the critical region; the second is the vanishing Seiberg--Witten central
	charge.  These are two aspects of the same quasi-homogeneous scaling,
	because \(\partial_t\lambda_C=\dd x/y\).  The lensing exponent agrees
	with the direct marginal-photon-sphere analysis of~\cite{Sasaki2025}.
	
	The distinction is summarized in Table~\ref{tab:fiber-collisions}.  The
	\(I_2\) point has mutually local vanishing cycles and remains
	semistable.  At the type \(II\) point the cycles have nonzero symplectic
	intersection, so mutually non-local BPS states become massless and the
	rank-one theory reaches an Argyres--Douglas fixed point
	\cite{ArgyresDouglas1995,ArgyresLotitoLuMartone2016}.
	
	\begin{table}[ht]
		\centering
		\small
		\begin{tabular}{@{}lll@{}}
			\toprule
			fiber & geometry and monodromy & lensing behavior \\
			\midrule
			\(I_1\) & one node; \(T\) & one logarithm \\
			\(I_2\) & two nodes; \(T^2\), mutually local
			& chain-dependent sum of logarithms \\
			\(II\) & one cusp; mutually non-local collision
			& \((t_{\rm AD}-t)^{-1/6}\) \\
			\bottomrule
		\end{tabular}
		\caption{The multiplicative (\(I_1\), \(I_2\)) and additive (\(II\)) degenerations
			relevant to the photon curve.}
		\label{tab:fiber-collisions}
	\end{table}
	
	The gravitational critical point can be stated invariantly.
	
	\begin{prop}[Degenerate photon orbit]
		\label{prop:degenerate-photon-orbit}
		A generic circular photon orbit is a double root
		\(q=q_x=0\), \(q_{xx}\ne0\), and gives the logarithmic
		strong-deflection law.  At~\eqref{eq:canonical-AD-st}, we have
		\[
		q=q_x=q_{xx}=0,\qquad q_{xxx}\ne0.
		\]
		The quadratic radial instability, and hence the corresponding Lyapunov
		exponent, vanishes; the local lensing scaling is instead the fractional
		power in~\eqref{eq:IItype}.
	\end{prop}
	
	The relation between circular null orbits and their Lyapunov exponents is
	standard~\cite{CardosoEtAl2009}.  For the asymptotically flat
	Reissner--Nordstr\"om spacetime, however, \(s=9/32>1/4\) lies beyond the
	black-hole charge bound.  Section~\ref{sec:Kiselev_horizons_AD} explains
	precisely when the same algebraic degeneration coincides with an
	ultracold horizon configuration.
	\par Comparing~\eqref{eq:E7dict} with the gravitational parameter $s=Q^2/4M^2$ gives
	\begin{equation}
		\tan^2m=-\frac{Q^2}{M^2}.
		\label{eq:tanmQM}
	\end{equation}
	For real gravitational charge the right-hand side of
	\eqref{eq:tanmQM} is negative, so the corresponding Wilson-line parameter
	is reached by analytic continuation.  The base coordinates are related
	algebraically by \(u=128\cos^2m/t\).  This inversion identifies the two
	families and their singular fibers; it should not be interpreted as a
	renormalization-group trajectory, since a Coulomb-branch modulus is a
	vacuum expectation value rather than an RG scale.  Likewise, the
	isomonodromic deformation in \(s\) is a statement about the
	Gauss--Manin connection, not about running of the MNW mass.
	Distinguished mass values reproduce the special rational elliptic
	surfaces found in Section~\ref{sec:surface}
	\cite{Herfurtner1991,Doran2001}; their configurations are
	listed in Table~\ref{tab:E7-special-values}.  In particular, the fiber over
	\(u=0\) is generically of type \(III^*\), and enhances to \(II^*\) at
	\(m=0\).
	
	\begin{table}[ht]
		\centering
		\begin{tabular}{@{}lll@{}}
			\toprule
			\(m\) & \((\sigma,s)\) & singular fibers and Mordell--Weil group \\
			\midrule
			\(0\)
			& \((0,0)\)
			& \(II^*+2I_1\), trivial \({\rm MW}\) \\
			\(i\infty\)
			& \((2/3,\,1/4)\)
			& \(III^*+I_2+I_1\), \({\rm MW}\cong\mathbb Z/2\mathbb Z\) \\
			\(\pm\pi/2\pm i\operatorname{arsinh}(2\sqrt2)\)
			& \((1,\,9/32)\)
			& \(III^*+II+I_1\), \({\rm MW}\cong A_1^*\) \\
			\bottomrule
		\end{tabular}
		\caption{Distinguished equal-mass members of the non-critical \(E_7\)
			family.}
		\label{tab:E7-special-values}
	\end{table}
	
	A multiple zero of the discriminant is not by itself sufficient to
	identify an Argyres--Douglas theory: the \(I_2\) collision at
	\(s=\tfrac14\) is multiplicative and its vanishing charges are mutually
	local.  The type \(II\) collision at
	\(s=\tfrac9{32}\), by contrast, is additive and is produced by mutually
	non-local cycles.  It therefore gives the Argyres--Douglas point of this
	one-parameter equal-mass slice
	\cite{ArgyresDouglas1995,ArgyresLotitoLuMartone2016}.  Under
	\eqref{eq:E7dict},
	\begin{equation}
		(s,t)=\left(\frac9{32},\frac8{27}\right)
		\quad\longleftrightarrow\quad
		u_{\rm AD}=-3456 .
		\label{eq:E7-AD-coordinate}
	\end{equation}
	The quasi-homogeneous scaling is precisely
	\eqref{eq:IItype}: the period of \(\lambda_C\) on the vanishing charge
	goes to zero as \((t_{\rm AD}-t)^{5/6}\), while its \(t\)-derivative,
	the relative holomorphic period governing the deflection angle, scales as
	\((t_{\rm AD}-t)^{-1/6}\).
	\section{The Kiselev deformation and the AD locus}
	\label{sec:Kiselev_horizons_AD}
	
	We now restore the linear Kiselev parameter \(a\) and the cosmological
	parameter \(\ell\) of~\eqref{eq:dimensionless}.  The name
	\emph{Kiselev deformation} is used in the sense explained in
	Section~\ref{sec:introduction}; see also~\cite{Kiselev2003,Visser2020}.
	Proposition~\ref{prop:linear-term-removal} gives a reduction to the family with no linear term.  Conversely, 
	we can take the point~\eqref{eq:E7-AD-coordinate} and obtain the
	corresponding point on the curve~\eqref{eq:photon-curve_quintessence}.
	\begin{prop}
		\label{prop:Kiselev-AD-locus}
		Away from the exceptional normalization locus, the inverse image of
		\[
		(\widetilde s,\widetilde t)
		=\left(\frac9{32},\frac8{27}\right)
		\]
		is the cuspidal locus, parametrized by \(R\in\mathbb C^\times\),
		\begin{equation}
			a_{\rm AD}=\frac{4R-3}{3R^2},
			\qquad
			s_{\rm AD}=\frac{R(3-R)}6,
			\qquad
			t_{\rm AD}=\frac{1-R}{2R^3},
			\qquad
			x_{\rm AD}=\frac1R.
			\label{eq:generalized-ultracold-dimensionless-AD}
		\end{equation}
		For generic \(R\) and  \((s,a)=(s_{\rm AD},a_{\rm AD})\) fixed, the one-parameter family has a
		Kodaira \(II\) fiber at \(t=t_{\rm AD}\).  At \(R=1\), the fiber remains of type \(II\), whereas at \(R=\tfrac32\)
		the fiber is of type \(III\). Moreover, for fixed parameter $a$ the Gauss--Manin connection of the family~\eqref{eq:photon-curve_quintessence}
		is obtained by pullback from~\eqref{eq:PFhom} and~\eqref{eq:GMeqn}, and
		the periods are related as follows:
		\begin{equation}
			\Pi_a(t;s,a)
			=
			c(t)\,
			\widetilde\Pi\bigl(\widetilde t(t),\widetilde s(t)\bigr),
			\qquad
			c(t)=\sqrt{\frac{8t}{D}}.
			\label{eq:period-pullback}
		\end{equation}
	\end{prop}
	\begin{proof}
		For the curve~\eqref{eq:photon-curve_quintessence} one has the binary-quartic invariants~\cite{Fisher2008}
		\begin{equation}
			\mathcal I=1-3a-12st,
			\qquad
			\mathcal J=2-9a+27sa^2+(72s-27)t.
			\label{eq:Kiselev-quartic-invariants}
		\end{equation}
		The Jacobian curve has
		\begin{equation}
			g_2^{(a)}=\frac43\mathcal I,
			\qquad
			g_3^{(a)}=\frac4{27}\mathcal J,
			\qquad
			\Delta^{(a)}
			=\frac{16}{27}\bigl(4\mathcal I^3-\mathcal J^2\bigr).
			\label{eq:Kiselev-Weierstrass-invariants}
		\end{equation}
		Under~\eqref{eq:linear-removal-map},
		\begin{equation}
			\widetilde{\mathcal I}
			=\left(\frac{8t}{D}\right)^2\mathcal I,
			\qquad
			\widetilde{\mathcal J}
			=\left(\frac{8t}{D}\right)^3\mathcal J.
			\label{eq:tilde-invariant-scaling}
		\end{equation}
		Equations~\eqref{eq:Kiselev-quartic-invariants}--\eqref{eq:tilde-invariant-scaling}
		both verify the formulas in~\eqref{eq:tilde-parameters} and show how the
		normalized AD point can be pulled back. At the normalized point both quartic invariants vanish.  By
		\eqref{eq:tilde-invariant-scaling}, its generic inverse image is
		\(\mathcal I=\mathcal J=0\).  Equivalently, imposing
		\(q=q_x=q_{xx}=0\) and writing the triple root as \(x=1/R\) gives
		\eqref{eq:generalized-ultracold-dimensionless-AD}.  Since
		\(q_{xxx}\ne0\) generically, the fiber is cuspidal and of type \(II\).
		At \(R=1\), \(t_{\rm AD}=0\) and the
		fractional-linear map itself is obtained only as a limit, although the
		fiber remains of type \(II\).  At
		\begin{equation}
			R=\frac32,
			\qquad
			(a,s,t)=\left(\frac49,\frac38,-\frac2{27}\right),
			\label{eq:exceptional-type-III}
		\end{equation}
		one has \(D=N=0\) and
		\(q=-\tfrac38(x-\tfrac23)^4\).  This is the exceptional type \(III\)
		member of the closure of the locus, not a generic AD cusp.
		\par If the cycles are transported by
		\eqref{eq:linear-removal-map}, then
		one obtains~\eqref{eq:period-pullback}. Using the no-\(a\) Gauss--Manin equation~\eqref{eq:GMeqn} gives the
		explicit pullback rule
		\begin{equation}
			\frac{\dd\Pi_a}{\dd t}
			=
			\left(
			\frac{\dot c}{c}
			-\dot{\widetilde s}\,
			B(\widetilde t,\widetilde s)
			\right)\Pi_a
			+
			c\left(
			\dot{\widetilde t}
			-\dot{\widetilde s}\,
			A(\widetilde t,\widetilde s)
			\right)
			\partial_{\widetilde t}\widetilde\Pi .
			\label{eq:Kiselev-GM-pullback}
		\end{equation}
		Here a dot denotes \(\dd/\dd t\) at fixed \((s,a)\), and the functions
		\(A,B\) are those in~\eqref{eq:ABcoeffs} evaluated at
		\((\widetilde t,\widetilde s)\).
		A second derivative, followed by \eqref{eq:PFhom}, produces the scalar
		Picard--Fuchs equation along the fixed-\((s,a)\) family wherever the
		pulled-back connection basis is nondegenerate.  
		One must not instead hold \(\widetilde s\) fixed: at fixed physical
		\((s,a)\),
		\begin{equation}
			\left.\frac{\dd}{\dd t}\right|_{s,a}
			=
			\dot{\widetilde t}\,\partial_{\widetilde t}
			+\dot{\widetilde s}\,\partial_{\widetilde s}.
			\label{eq:Kiselev-chain-rule}
		\end{equation}
		Thus the two-variable
		\(a=0\) Gauss--Manin system analyzes arbitrary-\(a\) periods by pullback.
	\end{proof}
	The fractional-linear change of \(x\) in Proposition~\ref{prop:Kiselev-AD-locus} need not preserve the positive real radial interval, so horizon
	ordering is best analyzed in the original variables.
	\subsection{Horizon bounds}
	With \(R=r/(2M)=1/x\), positive horizons are the positive roots of
	\begin{equation}
		P_{s,a,\ell}(R)
		=
		R^2 f(2MR)
		=
		s-R+R^2-aR^3-\ell R^4.
		\label{eq:horizon-polynomial-R}
	\end{equation}
	This polynomial is independent of the photon energy and angular momentum.
	A double horizon at \(R_0>0\) satisfies \(P=P'=0\), or
	\begin{equation}
		D_{a,\ell}(R_0)
		:=
		4\ell R_0^3+3aR_0^2-2R_0+1=0,
		\qquad
		s=s_{\rm ext}(R_0)
		:=\frac{R_0}{4}\bigl(3-2R_0+aR_0^2\bigr).
		\label{eq:extremal-horizon-parametrization}
	\end{equation}
	For \(a\ge0\) and \(\ell>0\), the cubic \(D_{a,\ell}\) has two
	positive critical radii \(R_-<R_+\) precisely when
	\begin{equation}
		d(a,\ell)
		=
		36a^2-108a^3+128\ell
		-432a\ell-432\ell^2>0.
		\label{eq:critical-cubic-discriminant}
	\end{equation}
	Here \(R_-\) is a local minimum and \(R_+\) a local maximum of \(P\).
	Consequently the three-horizon black-hole region is
	\begin{equation}
		\max\{0,s_{\rm ext}(R_+)\}
		<s<s_{\rm ext}(R_-).
		\label{eq:exact-black-hole-window}
	\end{equation}
	Equality at \(s=s_{\rm ext}(R_+)\) or
	\(s=s_{\rm ext}(R_-)\) gives the Nariai or cold boundary,
	respectively; the additional cutoff at \(s=0\) simply enforces
	nonnegative \(Q^2\).
	These are the standard two extremal branches of charged de~Sitter black
	holes~\cite{Romans1992,CastroMarianiToldo2023}.
	
	The intersection of the two branches is the ultracold locus, on which \(P\) has a triple
	positive root.  Solving \(P=P'=P''=0\) gives
	\begin{equation}
		a_{\rm uc}=\frac{4R_{\rm uc}-3}{3R_{\rm uc}^2},
		\qquad
		s_{\rm uc}=\frac{R_{\rm uc}(3-R_{\rm uc})}{6},
		\qquad
		\ell_{\rm uc}=\frac{1-R_{\rm uc}}{2R_{\rm uc}^3}.
		\label{eq:generalized-ultracold-dimensionless}
	\end{equation}
	The conditions \(a,s,\ell\ge0\) restrict
	\(3/4\le R_{\rm uc}\le1\).  At \(a=0\), one obtains the familiar
	ultracold Reissner--Nordstr\"om--de~Sitter point
	\begin{equation}
		R_{\rm uc}=\frac34,
		\qquad
		s_{\rm uc}=\frac9{32},
		\qquad
		\ell_{\rm uc}=\frac8{27}.
		\label{eq:RN-dS-ultracold-point}
	\end{equation}
	
	For completeness, at \(\ell=0\) and \(0<a<1/3\), put
	\(\delta=\sqrt{1-3a}\).  The two critical radii and the exact charge
	window reduce to
	\begin{equation}
		R_-=\frac1{1+\delta},
		\qquad
		R_+=\frac1{1-\delta},
		\qquad
		\max\left\{0,\frac{R_+(2-R_+)}3\right\}
		\le s\le
		\frac{R_-(2-R_-)}3.
		\label{eq:zero-Lambda-exact-window}
	\end{equation}
	Hence \(s\le1/4\) is the \(a=0\) Reissner--Nordstr\"om bound, not a
	universal bound after the Kiselev deformation.
	\par The agreement of the first two formulas in
	\eqref{eq:generalized-ultracold-dimensionless-AD} and
	\eqref{eq:generalized-ultracold-dimensionless} has an exact origin, namely
	\begin{equation}
		R^4q_{s,t,a}(1/R)
		=
		(t-\ell)R^4-P_{s,a,\ell}(R),
		\qquad
		t-\ell=\frac{4M^2E^2}{L^2}.
		\label{eq:reciprocal-photon-horizon-identity}
	\end{equation}
	We have the following:
	
	\begin{prop}
		\label{prop:zero-energy-intersection}
		Assume that we take the same real parameter \(R\) on the Kiselev Argyres--Douglas
		locus and on the physical ultracold branch, and fix \(L\ne0\).  Then the two parametrizations define the same
		point of the photon--horizon parameter space if and only if the Killing energy vanishes, \(E=0\), that is, \(t=\ell\).
	\end{prop}
	
	\begin{proof}
		The two loci have identical \(a(R)\) and \(s(R)\), while their remaining
		coordinates are respectively \(t_{\rm AD}(R)\) and
		\(\ell_{\rm uc}(R)\), given by the same rational function of \(R\).
		Equation~\eqref{eq:dimensionless} then gives
		\(t_{\rm AD}-\ell_{\rm uc}=4M^2E^2/L^2\), which vanishes exactly when
		\(E=0\).  The reciprocal identity
		\eqref{eq:reciprocal-photon-horizon-identity} identifies the triple roots.
	\end{proof}
	
	For \(3/4\le R<1\), substituting the common values
	\(a(R),s(R),t(R)=\ell(R)\) into~\eqref{eq:tilde-parameters} gives
	\begin{equation}
		(\widetilde s,\widetilde t)
		=\left(\frac9{32},\frac8{27}\right).
		\label{eq:ultracold-normalizes-to-AD}
	\end{equation}
	Thus the full physical Kiselev ultracold branch, viewed through its
	zero-energy photon curve, normalizes fiberwise to the canonical
	AD cusp with $a=0$.  The endpoint \(R=1\) is obtained by a limit because
	\(t=\ell=0\). 
	For \(a=0\) both loci pass through
	\((s,t,\ell)=(9/32,8/27,8/27)\), but only on the
	zero-Killing-energy boundary.  At fixed nonzero \(L\), this is the
	\(b=L/E\to\infty\) limit and is not a finite-energy scattering orbit.
	Conversely, the asymptotically flat AD curve at
	\((s,t,\ell)=(9/32,8/27,0)\) has finite
	\(b^2=27M^2/2\), but the corresponding Reissner--Nordstr\"om geometry is
	overcharged and has no event horizon.  
	
	\section{Discussion}
	\label{sec:discussion}
	
	The null-geodesic quartic relates the lensing problem to the geometry of a family of
	rational elliptic surfaces.  Its closed periods carry the homogeneous
	Gauss--Manin system, while the finite-distance angular integral is a lift
	of the normal function associated with
	\(D_{\rm lens}=2P_T-P_R-P_S\).  This identifies the geometric content of
	the two differential equations found for Reissner--Nordstr\"om lensing
	in~\cite{Sasaki2024} and separates it from the elementary finite-distance
	endpoint angles.
	
	The inhomogeneous equation has a Seiberg--Witten interpretation.
	The algebraic primitive \(\lambda_C=2y\,\dd x\) satisfies
	\(\dd\lambda_C=\dd t\wedge\dd x/y\), and
	Proposition~\ref{prop:horizontal-SW-gauge} constructs a horizontal gauge
	in which
	\[
	\mathcal L_t\lambda_{\rm h}=-G\,\dd t .
	\]
	Pairing this identity with the lensing chain, with the branch-point
	regularization of
	Proposition~\ref{prop:lensing-GM-system}, produces the inhomogeneous
	source.  Thus the source is not appended to the period geometry: it is
	the relative boundary of the Picard--Fuchs image of a
	Seiberg--Witten primitive.  The comparison with the MNW differential
	then shows exactly what the curve identification does and does not fix:
	the effective coupling is common, whereas the two-form normalization,
	special K\"ahler scale, and mass residues require the factors and
	third-kind correction derived in Section~\ref{sec:SW}.
	
	The two enhanced fibers have distinct meanings.  At \(s=1/4\), the
	\(I_2\) fiber is semistable and has mutually local vanishing cycles; it
	does not double the logarithm at the separate physical strong-field
	\(I_1\) fiber.  At \(s=9/32\), the type \(II\) cusp has mutually
	non-local vanishing cycles.  The Seiberg--Witten period vanishes as
	\(\epsilon^{5/6}\), while the differentiated relative period governing
	lensing scales as \(\epsilon^{-1/6}\).  This is the
	Argyres--Douglas degeneration of the equal-mass \(E_7\) slice, although
	in asymptotically flat Reissner--Nordstr\"om gravity it lies beyond the
	black-hole charge bound.
	
	Finally, the transformation \eqref{eq:tilde-parameters} reduces the quartic
	with arbitrary Kiselev parameter $a$ fiberwise to the family with $a=0$.  The full two-variable
	Gauss--Manin system pulls back by~\eqref{eq:Kiselev-GM-pullback}; a
	fixed-\(\widetilde s\) Picard--Fuchs equation does not.  The horizon
	problem must also retain the original real radial coordinate.  In those
	variables the AD and ultracold parametrizations coincide algebraically,
	but Proposition~\ref{prop:zero-energy-intersection} supplies the missing
	physical condition:
	\[
	t-\ell=\frac{4M^2E^2}{L^2}.
	\]
	At fixed \(L\ne0\), the two loci meet only at zero Killing energy.  The
	ultracold point is therefore a boundary of the photon-curve
	parameter space, not a finite-energy lensing configuration.
	
	It would be useful to extend the relative-cohomology construction to
	rotating black holes, where the separated radial and angular curves
	introduce additional moduli, and to determine whether the pullback method
	for the Kiselev term admits a comparably economical form-level
	Picard--Fuchs description.
	
	\bibliographystyle{unsrt}
	{\small
		\bibliography{References}
	}
	\appendix
	\section{Picard--Fuchs and form-level reduction coefficients}
	\label{app:coefficients}
	
	\subsection{Stiller's method for computing the Picard--Fuchs equation}
	\label{ss-stiller}
	
	Consider a rational elliptic surface $\pi : \mathcal{E} \to \mathbb{P}^1_{(t)}$, whose smooth fibers $E_t := \pi^{-1}(t)$ are represented by a Weierstrass model
	\begin{equation}
		\label{eq:gen_weier}
		E_{t}:\qquad \eta^2=4\xi^3-g_2(t)\,\xi-g_3(t) \, ,
	\end{equation}
	with non-vanishing holomorphic 1-form represented by the algebraic 1-form $\dd\xi/\eta$. For $\mathcal{E}$ a rational elliptic surface, the coefficient functions $g_2(t)$, $g_3(t)$ are polynomial functions of the affine coordinate $t$ of degree at most 4 and 6, respectively. As in \eqref{eq:weier}, these coefficient functions may depend on auxiliary parameters. For a 1-cycle $\gamma$ on the smooth fiber $E_t$, consider the period and the quasi-period
	\begin{equation}
		\label{eq:period_integrals}
		\Pi_\gamma := \oint_\gamma \frac{\dd\xi}{\eta} \, , \qquad
		\mathrm{H}_\gamma := \oint_\gamma \frac{\xi\,\dd\xi}{\eta} \, .
	\end{equation}
	Inherently, both depend on $t \in \mathbb{P}^1_{(t)}$. It is well known that the cycle $\gamma$ may be chosen locally constant in $t$, and after such a choice $\Pi_\gamma = \Pi_\gamma(t)$ and $\mathrm{H}_\gamma = \mathrm{H}_\gamma(t)$ become multivalued holomorphic functions that are defined on the complement of the singular locus of the fibration. Stiller \cite[\S 3]{Stiller1981} shows, following techniques of Griffiths \cite{MR229641}, that they satisfy the differential relations
	\begin{equation}
		\label{eq-elliptic_pf}
		\renewcommand{\arraystretch}{1.3}
		\frac{\dd}{\dd t}\left(\begin{array}{c}
			\Pi_\gamma \\
			\mathrm{H}_\gamma
		\end{array}\right)=
		\left(\begin{array}{cc}
			-\frac{1}{12} \frac{\Delta^{\prime}}{\Delta} & \frac{3}{2}\frac{ \delta}{ \Delta} \\
			-\frac{1}{8}\frac{g_2 \delta}{ \Delta} & \frac{1}{12} \frac{\Delta^{\prime}}{\Delta}
		\end{array}\right) \cdot\left(\begin{array}{c}
			\Pi_\gamma \\
			\mathrm{H}_\gamma
		\end{array}\right) \, ,
	\end{equation}
	where $\Delta(t) = g_2(t)^3-27g_3(t)^2$ is the discriminant of \eqref{eq:gen_weier} and $\delta(t) = 3g_3(t)g_2(t)^{\prime} - 2g_2(t)g_3(t)^{\prime}$. Eliminating $\mathrm{H}_\gamma(t)$, one obtains a second-order linear homogeneous equation satisfied by $\Pi_\gamma(t)$, that is, the Picard--Fuchs operator as in \eqref{eq:PFhom}. Equation \eqref{eq-elliptic_pf} is the Gauss--Manin connection of the elliptic pencil.
	
	\subsection{Coefficients of the Picard--Fuchs equation}
	\label{ss:coeffs}
	
	The polynomial numerators in~\eqref{eq:PQterm} are
	\begin{align}
		\mathcal N_P={}&
		4096s^5t^3-1536s^4t^3+2560s^4t^2-1920s^3t^2
		+512s^3t \notag\\
		&+648s^2t^2-704s^2t-81st^2+32s^2+252st-16s-27t+2,
		\label{eq:NP-appendix}\\
		\mathcal N_Q={}&
		1024s^5t^2-384s^4t^2+512s^4t-32s^3t+64s^3
		-48s^2t \notag\\
		&-112s^2+6st+44s-5.
		\label{eq:NQ-appendix}
	\end{align}
	The rational functions entering~\eqref{eq:KPF} are
	\begin{equation}
		R_1=-\frac12\,
		\frac{
			-s(4s-1)(48s^2t-18st-52s+15)x^2+P_1(s,t)\,x+P_0(s,t)
		}{
			t(16s^2t-6st+4s-1)(256s^3t^2+128s^2t-144st+16s+27t-4)
		},
	\end{equation}
	with
	\begin{align}
		P_1(s,t)&=1536s^5t^2-576s^4t^2+768s^4t+48s^3t+96s^3-132s^2t-272s^2+18st+122s-15,\\
		P_0(s,t)&=-384s^4t^2+144s^3t^2-288s^3t+96s^2t+80s^2-6st-40s+5,
	\end{align}
	and
	\begin{equation}
		R_3=-\frac12\,\frac{N_3}{D_3},
		\qquad
		\begin{aligned}
			N_3={}&-2s(16s^2t-6st+4s-1)x^3+(-3st+8s-2)x^2\\
			&+P_1^{(3)}(s,t)x+t(16s^2t-12s+3),\\
			D_3={}&t(256s^3t^2+128s^2t-144st+16s+27t-4),
		\end{aligned}
	\end{equation}
	where
	\(P_1^{(3)}(s,t)=-64s^3t^2-48s^2t+50st-8s-9t+2\).
	Both are used only through
	the combination $G=R_1K+R_3K^3$ of Section~\ref{sec:SW}, whose boundary
	values at $x=x_0(s,t)$ and $x=0$ give the closed-form
	inhomogeneities~\eqref{eq:GMinhom1}--\eqref{eq:GMinhom2}.

\end{document}